\documentclass[12pt]{elsarticle}
\pdfoutput=1

\usepackage{tikz}
\usetikzlibrary{shapes.geometric}
\usetikzlibrary{shapes.arrows}
 \usetikzlibrary{arrows,shapes,decorations.pathmorphing,backgrounds,positioning,fit,petri,automata}
\usepackage{array}

\usepackage[utf8]{inputenc}
\usepackage[english]{babel}
\usepackage{graphicx}
\usepackage{natbib}
\usepackage{amsfonts,amssymb,amsmath,amscd,amsthm,latexsym}
\usepackage{mathrsfs}
\usepackage{appendix}
\usepackage{algpseudocode}

\usepackage[ruled,vlined,linesnumbered]{algorithm2e}

\usepackage{enumitem}
\usepackage{txfonts}

\usepackage{fullpage}
\usepackage{booktabs}

\makeatletter
\def\ps@pprintTitle{%
 \let\@oddhead\@empty
 \let\@evenhead\@empty
 \def\@oddfoot{\footnotesize \it \hfill\today}%
 \let\@evenfoot\@oddfoot}
\makeatother

\renewcommand{\epsilon}{\varepsilon}

\newcommand{\esp}{\mathbb{E}}

\newcommand{\dd}{\mathrm{d}}

\newcommand{\ESS}{\operatorname{ESS}}
\newcommand{\KL}{\operatorname{KL}}

\newcommand{\calQ}{\mathscr{Q}}

\newcommand{\ds}{\displaystyle}

\newtheorem{theo}{Theorem}
\newtheorem{lem}[theo]{Lemma}
\newtheorem{pro}[theo]{Proposition}

\theoremstyle{definition}

\title{Tempered, Anti-trunctated, Multiple Importance Sampling}

\author[1]{Grégoire Aufort \fnref{LAM,I2M}}
\ead{gregoire.aufort@lam.fr}

\author[2]{Pierre Pudlo \fnref{I2M}\corref{cor1}}
\ead{pierre.pudlo@univ-amu.fr}

\author[1]{Denis Burgarella \fnref{LAM}}
\ead{denis.burgarella@lam.fr}

\cortext[cor1]{Corresponding author}
\fntext[LAM]{Aix Marseille Univ, CNRS, CNES, LAM, Marseille, France}
\fntext[I2M]{Aix Marseille Univ, CNRS, I2M, Marseille, France}

\affiliation[1]{organization={Aix Marseille Univ, CNRS, CNES, LAM},
  addressline={38 rue Frédéric Joliot-Curie},
  postcode={13388},
  city={Marseille CEDEX 13},
  country={France}}

\affiliation[2]{organization={Aix Marseille Univ, CNRS, I2M},
  addressline={29 rue Frédéric Joliot-Curie},
  postcode={13453},
  city={Marseille CEDEX 13},
  country={France}}

\date{\today}

\begin{document}

\begin{abstract}
  Importance sampling is a Monte Carlo method that introduces a proposal distribution to sample
  the space according to the target distribution. Yet calibration of the proposal distribution
  is essential to achieving efficiency, thus the resort to adaptive algorithms to tune this
  distribution. In the paper, we propose a new adpative importance sampling scheme, named
  Tempered Anti-truncated Adaptive Multiple Importance Sampling (TAMIS) algorithm. We combine a
  tempering scheme and a new nonlinear transformation of the weights
  we named anti-truncation.
  For efficiency, we were also concerned not to increase the number of evaluations of the
  target density. As a result, our proposal is an automatically tuned sequential algorithm that
  is robust to poor initial proposals, does not require gradient computations and scales well
  with the dimension.
\end{abstract}


\maketitle

\section{Introduction}

Importance sampling is a Monte Carlo method that predates Markov Chain Monte Carlo (MCMC).  It was
and is still used to sample distributions. importance sampling targets $\pi(x)$ with draws from the
proposal distribution $q(x)$. A draw $x$ is weighted with $\pi(x)/q(x)$ to correct the discrepancy
between $q$ and $\pi$. When $\pi \ll q$, these algorithms are unbiased. Moreover, when the density
of the target $\pi(x)$ is known up to a constant, we normalized the weights by their sum, which
introduce a small bias that has been well studied \citep[see, e.g.][]{robert1999monte}.  Unlike
MCMC, importance sampling is an embarrassingly parallel algorithm that can easily be distributed on
CPU cores or clusters.  Moreover, importance sampling does not require to sort the wheat from the
chaff by finding the limit of the warm-up or burn-in period. And, since it is not based on local
moves, it may be able to discover the different modes of the target.  It has therefore received a
recent interest, in particular when considering algorithms that calibrates the tuning parameters of
the algorithm to the target \citep{AIS_bugallo}.

The efficiency of importance sampling depends heavily on the choice of the proposal. Many adaptive
algorithms \citep{oh1992adaptive, oh1993integration} have been proposed to calibrate the proposal
based on past samples from the target. Thus a temporal dimension is introduced in these algorithms
to adapt the tuning parameters of the proposal distribution: at time $t$, draws $x$ are sampled
from a distribution $q_t(x) = q(x|\theta_t)$ whose parameter $\theta_t$ is adapted on past
results. However these algorithms suffer from numerical instability and sensibility to the first
proposal used at initialization. For instance, \citet[][Section 2.6]{liu2001monte} claimed that
such algorithms were unstable.  Indeed estimating large covariance matrices from weighted samples
can lead to ill-conditioned estimation problems \citep[see, e.g., ][]{ellaham:hal-02019041}.  And
\citet{AMIS} asserted that the initial distribution of their algorithm has a major impact on the
accuracy of adaptive algorithms. They talked about the ``what-you-get-is-what-you-see'' nature of
such algorithms: these methods have to guess which part of the space is charged by the target based
on points of this space that have been previously visited. Several schemes have been introduced to
initialise the first proposal distribution. The initialization method proposed by \citet[][Section
4]{AMIS} requires multidimensional simplex optimization, hence requires many evaluations of
$\pi(\theta)$ that are then discarded. On the other hand, \citet{beaujean2013initializing} runs a
complete Metropolis-Hastings algorithm that can miss several modes of the target since it is based
on local moves.

Numerical instability may come from the fact that the adaptive algorithm can be trapped
around a point of the space that better fits the target than previously visited
points. When such phenomenon occurs, the algorithm misses important parts of the core of
the target: the learnt proposal distribution becomes concentrated around this point, and
the rest of the space to sample is eliminated forever. When the space to sample is of
moderate or large dimension, numerical instability becomes a major problem. Many ideas
were proposed to tackle the issue including tempering and clipping \citep{AIS_bugallo}.
Tempering \citep{SAIS, korba} can be implemented as replacing the target $\pi(x)$ by $\pi(x)^\beta$,
with $\beta<1$. It eases the discovery of the core of the target since it extends
the part of the space that is charged by the target. Thus, tempering can smooth the bridge
from the first proposal $q_1(x)$ to the target $\pi(x)$.
Clipping \citep{truncated,Koblents,vehtari2021pareto} of the importance weights is a non linear
transformation of the weights that decreases the importance of points with high
$w(x) = \pi(x)/q_t(x)$. The most common way to implement clipping as a variance
reduction method (which introduce a bias) is the truncation that deals with the degeneracy as
follows.  If $w(x)>S$ where $S$ is a threshold that needs to be calibrated, the weights $w(x)$ are
replaced by some value (e.g., by $S$). Otherwise, they are left unchanged. As noted by
\citet{Koblents} and \citet{AIS_bugallo}, this transformation of the weights flattens the target
distribution. Therefore, truncation is redundant with tempering.
Finally, in order to increase computational efficiency, schemes have been introduced to recycle
the successive samples generated at every iterations. In this vein, \citet{AMIS,MAMIS}
considered the whole set of draws from the different proposals calibrated at each stage of the
algorithm as drawn from a mixture of these distributions to significantly increase their
efficiency.

In this paper, we propose an adaptive importance sampling whose sensitivity to the first proposal,
and numerical instability are highly reduced. We have tried to design our algorithm to keep control
on the number of evaluations of the (unnormalized) target density.  In many situations where we are
interested in sampling the posterior distribution, the target density is indeed a complex function
of the paramaters $x$ and the data. For instance, an extreme case is a Gaussian model whose average
$\mu(x)$ is a blackbox function which carries a physical model of the reality given the value of
the parameters $x$.  Thus, the time complexity of our algorithm should be assessed in number of
evaluations of the proposal density.
We relied on a simple form of tempering to adapt the proposal distribution. Nevertheless tempering
was not enough to stabilize the algorithm on spaces of large dimension. In our algorithm, at each
stage after initialization, we update the proposal using tempering, and an anti-truncation that
replaces all weights $w(x)$ lower than a threshold $s$ by $s$. We show that this kind of clipping
can be considered as a contamination of the current proposal with the previous one.  As exhibited in
the numerical simulations in the last Section, both tricks (tempering and contamination with
previous proposal) avoid focusing too quickly on the few points with high weights. At
least, our method keeps the variance of the proposal large enough to take time to explore the space to
sample before exploiting the points with high importance weights.

\section{Calibration of importance sampling}
\label{sec:calibration}

We propose here a new strategie to walk on the bridge from the first proposal $q_1(x)$ to a
proposal $q_T(x)$ well adapted to the target $\pi(x)$ in terms of effective sample size.  In
order to adapt the proposal gradually, we introduce a sequence of temporary targets:
\[
  \widehat \pi_1(x), \ldots, \widehat \pi_T(x)
\]
which are intermediaries between the first proposal $q_1(x)$ and the target
$\pi(x)$. 
The precise definition of these temporary targets, given in
Section~\ref{sec:temporary-targets}, is paramount to the succes of the algorithm. They are based
on a tempering $w^\beta$ of the importance weights $w$. As described in
Section~\ref{sec:tempering}, the tempering 
\begin{enumerate}[label=\emph{(\roman*)}, noitemsep]
\item eases the discovery of the area charged by the real target $\pi(x)$,
\item temporarily removes the problems due to large queues of the target,
\item allows us to design a diagnostic based on the final of $\beta$.
\end{enumerate}
To this non-linear transformation of the weights, we add an anti-truncation, defined as
$\widehat w^\beta=w^\beta \vee s$, that pulls up all tempered weights $w^\beta$ less than a
threshold $s$ to this single value, see Section~\ref{sec:temporary-targets}.
This anti-truncation 
\begin{enumerate}[resume, label=\emph{(\roman*)}, noitemsep]
\item performs a contamination of the temporary target by the last proposal,
\item helps to stabilize numerically the algorithm and
\item allows us to explore new directions in large-dimension spaces.
\end{enumerate}
Both $\beta$ and $s$ are automatically calibrated at the end of stage $t$ of the algorithm, as
explained in Section~\ref{sec:auto-calibration}.  The new proposal $q_{t+1}(x)$ is tuned to fit
the temporary $\widehat\pi_t(x)$ with the EM algorithm as given in Section~\ref{sec:updating}.
The whole algorithm is given in Figure~\ref{fig:algorithm}.

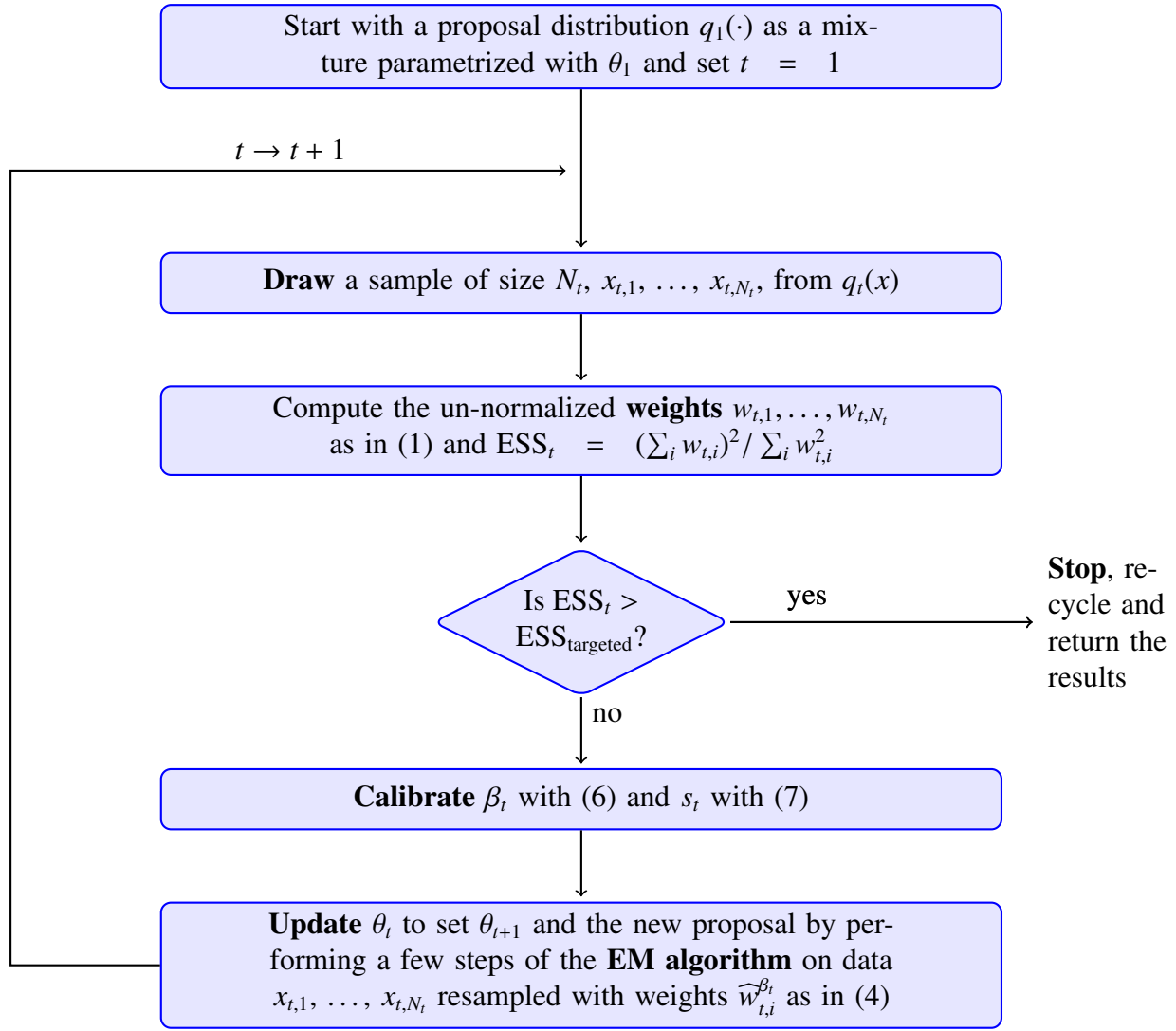
\begin{figure}
  \centering
   \begin{tikzpicture}[auto,
  decision/.style = { diamond, aspect=2, draw=blue, thick, fill=blue!10,
    text width=5em, text badly centered,
    inner sep=1pt, rounded corners },
  block/.style    = { rectangle, draw=blue, thick, 
    fill=blue!10, text width=11.5cm, text centered,
    rounded corners, minimum height=2em },
  line/.style     = { draw, thick, ->, shorten >=2pt },
  ]
  \matrix [column sep = 5mm, row sep = 10mm] { %
    & \node [block] (start) {Start with a proposal distribution $q_1(\cdot)$ as a mixture
      parametrized with $\theta_1$ and set $t=1$}; &
    \\
    & \node (debutboucle) {}; &
    \\
    & \node [block] (draw) {\textbf{Draw} a sample of size $N_t$, $x_{t, 1}$, \ldots,
      $x_{t,N_t}$, from $q_t(x)$\\}; &
    \\
    & \node [block] (weight) {Compute the un-normalized \textbf{weights}
      $w_{t,1},\ldots, w_{t,N_t}$ as in \eqref{eq:weights} and
      $\ESS_t = (\sum_i w_{t,i})^2/\sum_iw_{t,i}^2$}; &
    \\
    \node(null3){}; & \node [decision] (ESS) {Is $\text{ESS}_t>\text{ESS}_\text{targeted}$?}; &
    \node [text width=5em] (stop) {\textbf{Stop}, recycle and return the results};
    \\
    & \node [block] (calibrate) {\textbf{Calibrate} $\beta_t$ with \eqref{eq:calibrate_beta}
      and $s_t$ with \eqref{eq:s}}; &
    \\
    & \node [block] (EM) {\textbf{Update} $\theta_{t}$ to set
      $\theta_{t+1}$ and the new proposal by
      performing a few steps of the
      \textbf{EM algorithm} on data $x_{t,1}$, \ldots,
      $x_{t,N_t}$ resampled with weights $\widehat w_{t,i}^{\beta_t}$ as in
      \eqref{eq:hat_w} 
    }; &
    \\
  };

  \begin{scope} [every path/.style=line]
    \path (start)        --    (draw);
    \path (draw)      --    (weight);
    \path (weight)   --   (ESS);
    \path (ESS) -- node [near start] {yes} (stop);
    \path (ESS) -- node [near start] {no} (calibrate);
    \path (calibrate) -- (EM);
    \path (EM)   --++  (-8,0)  |- (debutboucle) node [near end] {$t\to t+1$};
    \path (ESS)   --    node [near start] {yes} (stop);
  \end{scope}
\end{tikzpicture}

  \caption{\label{fig:algorithm} {\bf The tempered, anti-truncated multiple importance sampling
      (TAMIS) algorithm} 
  }
\end{figure}

\subsection{The tempering}
\label{sec:tempering}

Let us assume that, given all past draws, a set
$x_{t,1}, \ldots, x_{t,N_t}$ of size $N_t$ has been drawn independently from a
distribution $q_t(x)=q(x|\theta_t)$ picked among a parametric family $\mathscr Q$ of laws. The
importance weights at this stage are
\begin{equation}
  w_{t,i}=\frac{\pi(x_{t,i})}{q_t(x_{t,i})}.\label{eq:weights}
\end{equation}

We can replace the target $\pi(x)$ by the distribution of density
\begin{equation}
  \label{eq:tempering}
  \pi_{\beta,t}(x) \propto \pi(x)^\beta q_t(x)^{1-\beta}
\end{equation}
with inverse temperature $\beta\in (0,1)$ as proposed by \citet{Neal} in his Annealed importance
sampling.  When $\beta=0$, \eqref{eq:tempering} is the proposal distribution that served to draw
the $x_{t,i}$'s: $\pi_{\beta=0,t}(x)=q_t(x)$. When $\beta=1$, \eqref{eq:tempering}
is the target distribution: $\pi_{\beta=1,t}(x) = \pi(x)$.  Moreover,
$\beta\mapsto \KL(\pi\|\pi_{\beta,t})$ decreases from $\KL(\pi|q_t)$ to $0$, see
Proposition~\ref{pro:KLbeta} in \ref{app:tempering}.
If we use the $x_{t,i}$'s to target $\pi_{\beta, t}(x)$, the unnormalized
importance weights become
\begin{equation}
  \label{eq:new_weight}
  \frac{\pi_{\beta,t}(x_{t,i})}{q_t(x_{t,i})} \propto
  \frac{ \pi(x)^{\beta} q_t(x)^{1-\beta}}{q_t(x_{t,i})}
  =
  \left(\frac{ \pi(x)}{q_t(x_{t,i})}\right)^\beta = w_{t,i}^\beta.
\end{equation}
Such weights have been use in the past, for instance by \citet{Koblents} who relied on the
$x_{t,i}$'s weighted with the $w_{t,i}^\beta$'s to get a sample from $\pi_{\beta,t}(x)$ and to
tune a $q_{t+1}(x)=q(x|\theta_{t+1})$ that approximates $\pi_{\beta,t}(x)$. It is also explored
by \citet{korba} as a regularization strategy.


\subsection{Anti-trunctation and temporary targets}
\label{sec:temporary-targets}

There are many ways to contaminate this weighted sample with draws from $q_t(x)$. The
first idea is to add $N_t'$ new draws $x_{t,N_t+1},x_{t,N_t+2}\ldots$ with all weights
equal to $s$ to the above weighted sample. This idea may add a non negligeable amount of
computational time when the dimension of $x$ is large.
Another idea to contaminate this weighted sample with $q_t(x)$, is to change the
weights. 
We introduce a deterministic contamination based on the value of $w_{t,i}^\beta$. Indeed, the
$x_{t,i}$'s weighted with
\begin{equation}
  \label{eq:hat_w}
  \widehat w_{t,i}^\beta = s \vee w_{t,i}^\beta 
\end{equation}
form an approximation of the distribution with density
\begin{equation}
  \label{eq:mixture}
  \widehat \pi_{\beta,t}(x) \propto s q_t(x) \mathbf 1\{x \in E\} +
  \pi^\beta(x)q_t^{1-\beta}(x) \mathbf 1\{x \not\in E\},
  \quad \text{where }
  E = \{x:\ \pi^\beta(x) / q_t^\beta(x) \le s\}.
\end{equation}
An easy computation gives us the weights of the mixture as follows.
\begin{lem}
  Let $q_t^E(x)$ denotes the normalized probability density of $q_t(x)$ knowing
  $x\in E$ and $\pi_{\beta,t}^{\bar E}(x)$ the normalized probability density of
  $\pi_{\beta,t}(x) = \pi^\beta(x)q_t^{1-\beta}(x)$ knowing $x\not\in
  E$.

  We have
  \[
    \widehat \pi_{\beta,t}(x) = \lambda q_t^E(x) + (1-\lambda) \pi_{\beta,t}^{\bar E}(x)
  \]
  where $\lambda =s \int_E q_t(x)\dd x = 1 - \int_{\bar E} \pi^\beta(x)
  q_t^{1-\beta}(x) \dd x$.
\end{lem}

Note that the scheme is different from the Safe Importance Sampling one \citep{SAIS,safe} as the anti-truncation contaminates the target with the current proposal $q_t$ instead of  $q_0$, and specifically in $E$. We apply in \eqref{eq:hat_w} a non-linear transformation of the weights. Yet it is the
inverse of truncating the importance weights and we refer to these transformed weights as
anti-truncated weights. Unlike the common truncation of the weights that replaces all
weights larger than $S$ by $S$, the anti-truncation we propose in \eqref{eq:hat_w} replaces
all weights smaller than $s$ by $s$.  Actually, we do not need to truncate large values
since we relied on tempering to remove the degeneracy of the
weights. However the sample drawn from $q_t(x)$ with weights $w_{t,i}^\beta$ may not
be of sufficient size to approximate \eqref{eq:tempering} correctly, even if $\beta$ is
well calibrated. If we trust that $q_t(x)$ is a decent sampling distribution, the anti-truncated,
tempered weights fight against the degeneracy of the weights in importance sampling (tempering) and
keep part of the old proposal ($q_t$) to keep exploring the space from it (anti-truncation).
At the end of each stage $t$ (except the final one), the future proposal distribution
$q_{t+1}(x)=q(x|\theta_{t+1})$ is calibrated on the temporary target given by
\eqref{eq:mixture}. The anti-truncated, tempered $\widehat\pi_{\beta,t}(x)$ defined in
\eqref{eq:mixture}, is a continuous bridge from
\begin{itemize}
\item the real target $\pi(x)$ to
\item the freshly used proposal $q_t(x)=q(x|\theta_t)$.
\end{itemize}
The tempered target $\pi_{\beta,t}(x)=\pi^\beta(x)q_t^{1-\beta}(x)$ is
already such a continuous bridge. But, when $\beta$ is fixed, the anti-truncated, tempered
$\widehat \pi_{\beta,t}$ is in-between the tempered $\pi_{\beta,t}$ and the freshly used
proposal $q_t(x)$ in terms of Kullback divergence as given by
Proposition~\ref{pro:KL_hat_pi}.  Let us recall first that, if both $f$ and $g$ are
probability densities, then the Kullback divergence is defined as
\[
  \KL(f||g) = \int f(x) \log \frac{f(x)}{g(x)} \dd x.
\]
If $f$ and $g$ are unnormalized probability densities, we will still denote by $\KL(f||g)$
the Kullback divergence between their normalized versions.

The following proposition is proved in \ref{app:proof}.
\begin{pro}\label{pro:KL_hat_pi}
  When $s\le 1$, we have
  \[
    0 = \KL\big(\pi_{\beta,t} \big\| \pi_{\beta,t}\big) \le
    \KL\big(\pi_{\beta,t} \big\| \widehat\pi_{\beta,t}\big) \le
    \KL\big(\pi_{\beta,t} \big\| q_t\big)
  \]
\end{pro}

\subsection{Updating the proposal}
\label{sec:updating}

The family of proposals we recommend for TAMIS is composed of Gaussian
mixture models, with diagonal covariance matrix for each
component. The density of a distribution $q(x|\theta)\in \calQ$ is defined as
\[
  q(x|\theta) = \sum_{k=1}^K \mathfrak{p}_k\, \varphi(x|\mu_k, \Sigma_k)
\]
where $\varphi(x|\mu, \Sigma)$ is the multivariate Gaussian density
with mean $\mu$ and covariance matrix $\Sigma$. This family is
parametrized by
$\theta=(\mathfrak{p}_1,\ldots, \mathfrak{p}_K, \mu_1, \ldots, \mu_K,
\Sigma_1, \ldots, \Sigma_K)$.

The future proposal distribution $q_{t+1}(x)\in \calQ$ is set by using the EM
algorithm. Let us assume that $q_t(x) = q(x|\theta_t)\in \calQ$ is the
Gaussian mixture with
parameter $\theta_t$. We tune $q_{t+1}(x)=q(x|\theta_{t+1}) \in \calQ$, that is
to say, we pick $\theta_{t+1}$ with the help of the $x_{t,i}$'s weighted with $\widehat
w_{t,i}^\beta$ as given in \eqref{eq:hat_w}. After resampling this
sample occording to their weights $\widehat w_{t,i}^\beta$,
we resort to iterations of the EM algorithm, starting from $\theta_t$, to
get $\theta_{t+1}$. Because of well known properties of the EM
algorithm \citep[see, e.g., ][]{fruhwirth2019handbook}, we have that
\[
  \KL\big(\widehat \pi_{\beta,t} \big\| q_{t+1} \big) <
  \KL\big(\widehat \pi_{\beta,t} \big\| q_{t} \big).
\]

\section{Practical aspects of the TAMIS algorithm}
\label{sec:practical}

We can now discuss pratical aspects of the proposed algorithm, based
on numerical results that demonstrate the typical behavior of the
method. 

\subsection{Choosing the inverse temperature $\beta$ and the anti-truncation $s$}
\label{sec:auto-calibration}
The inverse temperature $\beta$ has to be chosen at each stage of the algorithm (except
the last one). We follow the path open by by \citet{Beskos} to chose $\beta$. To ensure
that the $x_{t,i}$'s weighted with $\widehat w_{t,i}^\beta$ is a sample that can
approximate $\widehat\pi_{\beta,t}$, we set $\beta$ automatically at each stage with
\begin{equation}\label{eq:calibrate_beta}
  \beta_t = \sup\big\{\beta\in(0,1):\ \ESS(\beta)>\ESS_\text{min}\big\}, \quad \text{where }
  \ESS(\beta) = {\ds \left(\sum_{i=1}^{N_t} w_{t,i}^\beta\right)^2}\bigg/{\ds
    \sum_{i=1}^{N_t} w_{t,i}^{2\beta}}.
\end{equation}
The function $\beta \mapsto \text{ESS}(\beta)$ is continuously decreasing (see
Proposition~\ref{decreasing_ESS} of the Appendix). Hence the optimization problem stated in
\eqref{eq:calibrate_beta} can be solved easily by a simple one-dimensional bisection method and do not require a new sampling step, contrary to \citet{korba}'s adaptive regularization scheme.
Note that the weights $\widehat w^\beta$ related to the temporary target \eqref{eq:mixture} are
used only to calibrate the next proposal $q_{t+1}(x)$ --- this is an important difference
with the algorithm proposed by \citet{Koblents}. Hence the value of $\ESS_\text{min}$ should be
fixed such that the fit of $q_{t+1}(x)$ with the EM algorithm provides stable estimates
with an iid sample of size $\ESS_\text{min}$.

A good choice of $\ESS_\text{min}$ is essential to get numerical stability in our
algorithm. If $\ESS_\text{min}$ is much larger than really needed, the
algorithm will remain stable numerically. But convergence to the target will be slow down:
as the tempering will be more aggressive at each stage, more iterations will be needed to
move from the first proposal $q_1(x)$ to the target $\pi(x)$. The typical effect
of changing the value of $\ESS_\text{min}$ is studied in Figure~\ref{fig:Delta}. For
example, if $\calQ$ is the set of mixtures of $K$ Gaussian densities with diagonal
covariances, the update of the proposal with EM steps require to calibrate $Kd$ mean
paramaters and $Kd$ variance parameters. Thus, we should have $2Kd \ll \ESS_\text{min} \le
N_t$.

The value of $s$ that set the amount of anti-truncation is more easy to tune. We choose $s$
to be the quantile of order $\tau$ of the tempered weights:
\begin{equation}
  \label{eq:s}
  s_t = \operatorname{quantile}_{\text{order}=\tau}\Big(w_{t,1}^\beta,\ldots, w_{t,N_t}^\beta\Big).
\end{equation}
Although the required number of iterations may be suboptimal, the value $\tau=0.4$ appears to be a
universal compromise, working flawlessly in every numerical example considered in this paper. Lower
values of $\tau$ picked in $(0,0.1)$ can speed up the algorithm in low dimensional problems, but can
induce instability. Hence, we strongly advocate for the almost universal $\tau = 0.4$, see
Figure~\ref{fig:tau}.

\subsection{Numerical diagnostics}
\label{sec:diagnostics}

In order to assess the convergence of the algorithm we monitor the inverse temperature and the
estimated Kullback-Leibler divergence along iterations. Following \cite{Capp__2008}, we estimate the
Kullback-Leibler divergence between the target density and the mixture proposal using the Shannon
entropy of the normalised IS weights. Indeed since the normalised perplexity $\exp(H^{t,N})/N$ is a
consistent estimator of $\exp(-KL(\pi ||q_t))$,where $H^{t,N} = -\sum_{i=1}^{N_t} \omega_{i,t} \log
\omega_{i,t}$ \citep{Capp__2008}, we simply estimate $KL(\pi ||q_t) \approx \sum_{i=1}^{N_t}
\omega_{i,t} \log \omega_{i,t} +\log N_t$. Note that this estimate is upper bounded by $\log N_t$,
leading to an obvious bias when $KL(\pi ||q_t)$ is large or $N_t$ small. However this bias
does not practically prevent the use of this estimate as a monitoring tool.

We show in Figure~\ref{fig:monitoring} the typical evolution of both the inverse temperature $\beta$
and the estimated KL divergence along iteration. The inverse temperature starts increasing slowly
during the first iterations, followed by a strong acceleration until it stabilises. The estimated KL
divergence on the other hand starts with a plateau at its upper bound ($\log N_t$), then drops to a
much small value as $\beta$ reaches it maximum.

In some cases, $\beta$ does not reach 1, nor does the estimated KL divergence reach 0. Indeed if
the target density can't be well approximated by any proposal in $\calQ$,
$KL(\pi ||q_t)$ never reaches 0. This behaviour is also observed on targets of very high dimension
regardless of the proposal distribution family (see Section~\ref{sec:dim}). 
Even in those pathological cases, the convergence of TAMIS can be simply assessed by the sharp
increase of $\beta$ followed by its stabilization (or the sharp decrease of the estimated KL).

\subsection{Stopping criterion and recycling}
\label{sec:stop}

When the iterative algorithm is stopped at time $T$, we end with a set of weighted simulations:
\[
  x_{t,i}\sim q_t(\cdot)=q(\cdot|\theta_t), \quad \text{with weight } w_{t,i}=\frac{\pi(x_{t,i})}{q_t(x_{t,i})}.
\]
As in many iterative importance sampling algorithms such as AMIS~\cite{AMIS}, we recycle all these
draws and change their weights to
\[
  w_{t,i} = \frac{\pi(x_{t,i})}{Q(x_{t,i})}, \quad \text{where } Q(x) = \frac{1}{N_1+\cdots +
    N_T}\sum_{t=1}^T N_t q_t(x).
\]

We use the usual effective sample size estimate to assess the quality of the IS sample given by
TAMIS. Thus we suggest stopping the algorithm when the predifined ESS or the maximal
number of iterations is reached.
As usual in such adaptive algorithms, we recycle all particles with their weights after stopping the
iterations.  This recycling improve the efficiency of the algorithm. Thus, the ESS of the
final sample returned by the algorithm is underestimated by the sum of the effective sample sizes at
each iteration. Hence, to monitor that we have reached the predefined level, we stop at the first
time where
\[
  \ESS_1+\cdots+\ESS_t > \ESS_\text{predefined}
\]
or when we reach the maximal number of iterations.

\begin{table}[tb]
  \centering
  \caption{Parameter tuning and monitoring experiments}
  \label{tab:exp3}
  \renewcommand{\arraystretch}{1.2}
  \begin{tabular*}{.95\textwidth}{@{\extracolsep{\fill}} l l l l}\toprule
    \textbf{Experiment} & E3.1 & E3.2 & E3.3
    \\ \hline
    \textbf{Dimension} & $d=50$ & $d=50$ & $d=1,000$
    \\
    \textbf{Target} & $\mathcal N(50, 5)^{\otimes d}$ &$\mathcal N(50, 5)^{\otimes d}$ &$\mathcal N(10, 5)^{\otimes d}$  
    \\
    \textbf{Proposals} & \multicolumn{2}{c}{Gaussian mixture with 5 components} & Gaussian
   \\
    \textbf{Draws} & $N_t=2,000$ & $N_t=2,000$ & $N_t=2,000$
    \\
    $\textbf{ESS}_\textbf{min}$ & $\in\{100,200, 1400\}$ & $300$ & $1,000$
    \\
    $\boldsymbol\tau$ & $0$ & $\in\{0, 0.1, \ldots, 0.9, 0.95\}$ & $0.4$
    \\
    \textbf{Stop} & \multicolumn{2}{c}{$\sum_t \ESS_t > 10,000$} & $t=500$
    \\ \bottomrule
  \end{tabular*}
\end{table}

\subsection{Parameter tuning and monitoring}

We start by illustrating the effect of parameter tuning on TAMIS with the experiments targeting
various multivariate Gaussian distribution as given in Table~\ref{tab:exp3}. The proposal at first
iteration was a Gaussian mixture with 5 components: each component is centered around a $\mu_k$
drawn at random from $\mathcal{U}([-4,4])^{\otimes d}$ and has covariance matrix $\Sigma_k=200 \times
\mathbb{I}_{50}$ with large eigenvalues. To approximate de MSE, we ran $20$ replicates of the
experiences for each set of parameters.

Figure~\ref{fig:Delta} shows Experiment E3.1 described in Figure~\ref{tab:exp3} and Figure~\ref{fig:tau} shows 
Experiment E3.2. The conclusion is that we should set $\ESS_\text{min}$ so that the
calibration of the new proposal (i.e., of $\theta_{t+1}$) is stable and that $\tau=0.4$ is a decent value.

To illustrate monitoring in Figure~\ref{fig:monitoring}, we first plot a typical tempering path
(obtained on Experiment E3.1 with $\ESS_\text{min} = 100$ and $\tau = 0$) along with the estimated
KL divergence. As mentioned in section \ref{sec:auto-calibration}, the auto calibrated tempering path
has a rather sigmoid-like shape with a clear transition and stabilization to $\beta = 1$, while
the KL-divergence decreases (despite the estimator bias at the beginning) until both quantities
stabilizes together around $1$ and $0$ respectively. 

Finally we illustrate the typical behavior of the monitoring on targets of very high dimension with
Experiment E3.3. The first proposal distribution to initialize TAMIS is a Gaussian distribution
centered at $\mu$ drawn from $\mathcal{U}([-4,4])^{\otimes d}$ and with covariance matrix
$\Sigma = 100 \times \mathbb{I}_{d}$. Figure~\ref{fig:high-dim} shows that TAMIS provide more than
decent results in high dimension.
\begin{figure}
\centering
    \includegraphics{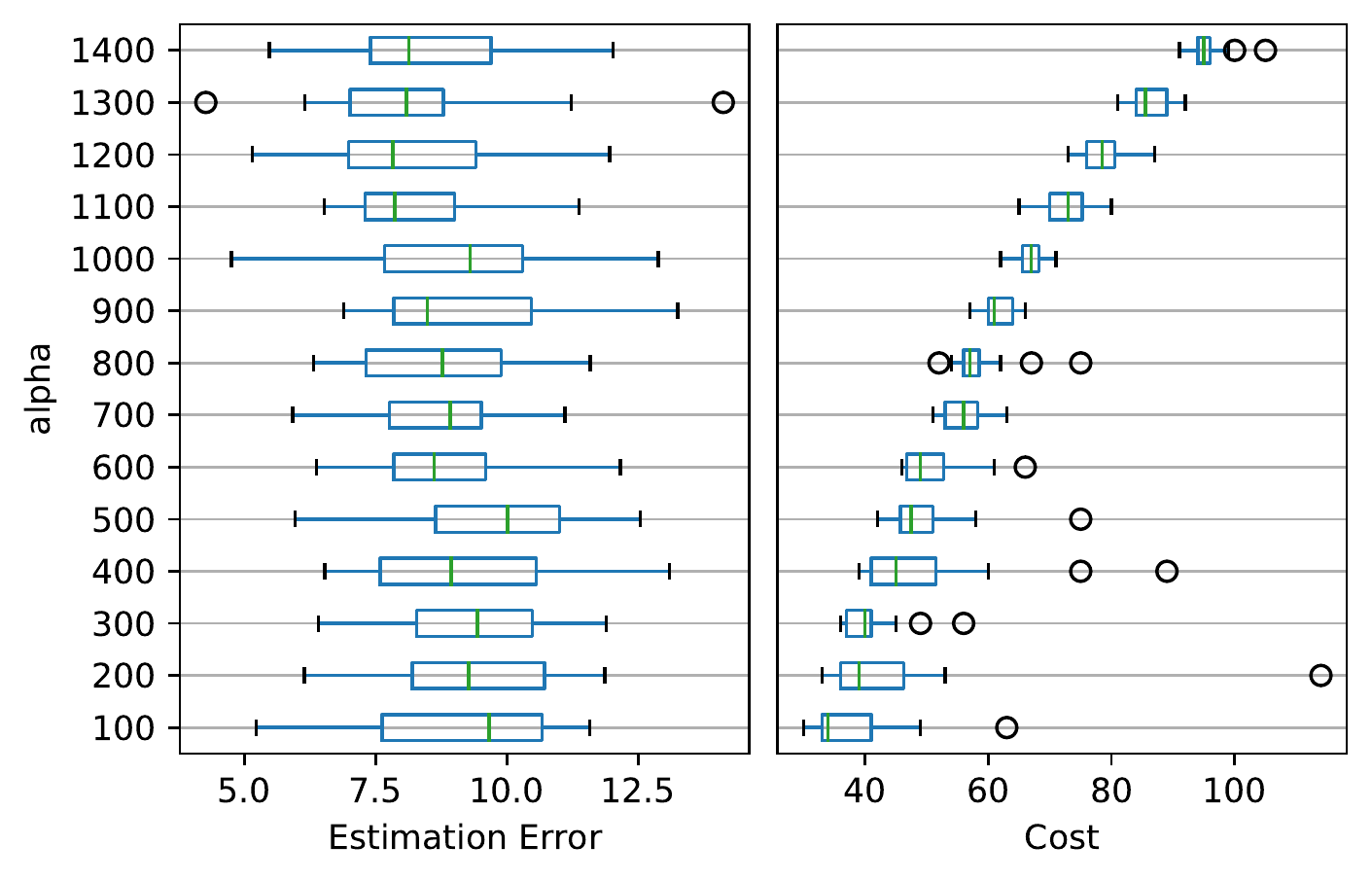}

    \caption{Effect of varying the $\ESS_\text{min}$ parameter ($y$-axis) defining the minimum ESS to be
      reached for calibration of the interse temperature. As stopping depends on the total estimated
      ESS, the MSE of the variance ($x$-axis on the left) estimation doesn't depend on $\ESS_\text{min}$, but the
      number of required iterations ($x$-axis on the right) before convergence of the sequence of proposal distributions
      increases. Increasing $\ESS_\text{min}$ further than the minimum required to stabilize
      the calibration of the new proposal (i.e., of $\theta_{t+1})$ with the EM step results in an increased computational cost.
    }
    \label{fig:Delta}
\end{figure}

\begin{figure}
\centering
   \includegraphics{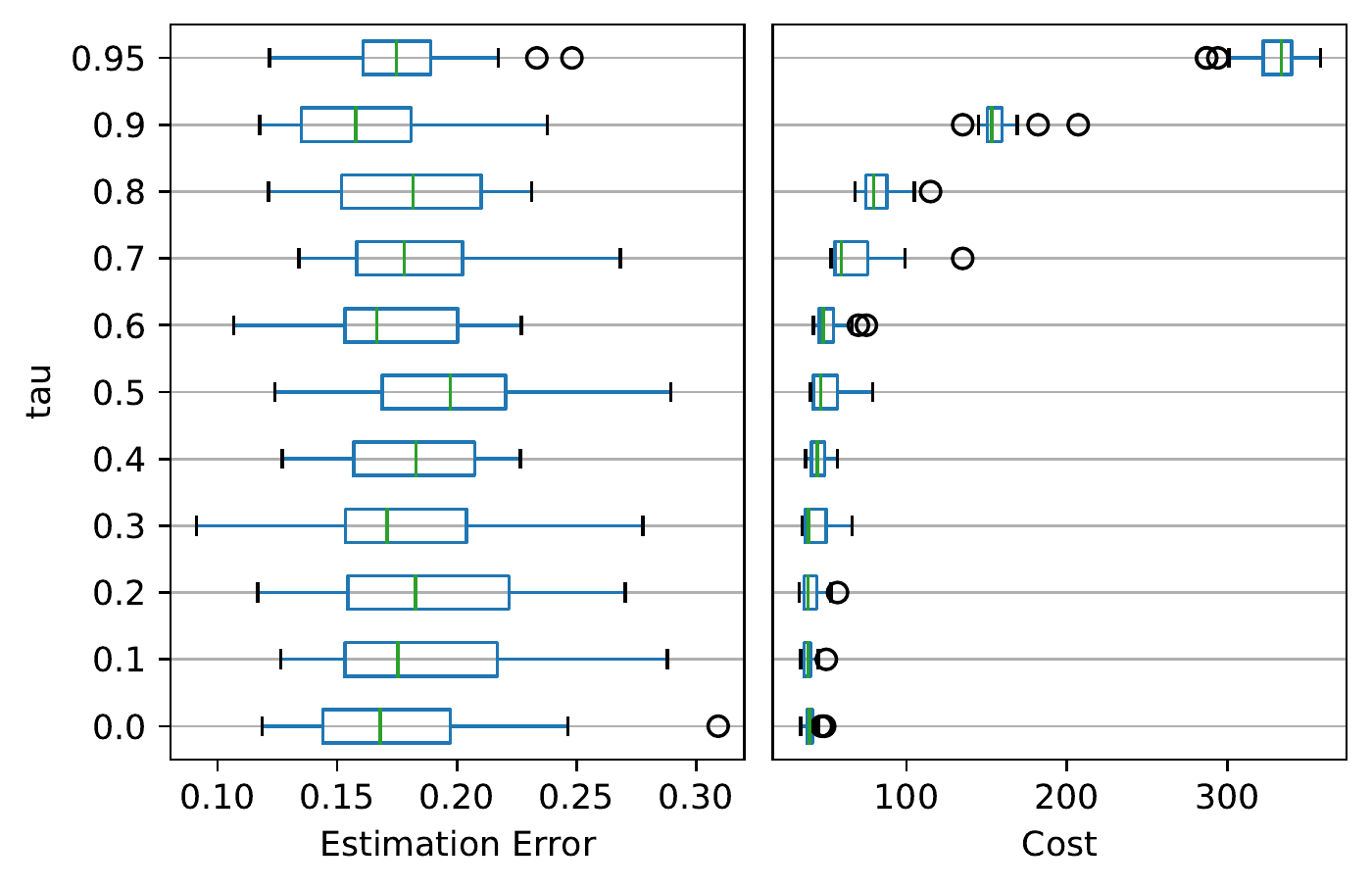}
   \caption{Effect of varying the $\tau$ parameter ($y$-axis) defining the antitruncation
     threshold. Except for very high values, the truncation has no detrimental effect on either the
     MSE ($x$-axis on the left) of the estimated variance or the required number of iterations
     ($x$-axis on the right) before convergence of the sequence of proposal distributions
     increases. As for $\ESS_\text{min}$, once the calibration of of the new proposal (i.e., of
     $\theta_{t+1})$ with the EM step is stable, increasing $\tau$ further only increases the
     computational cost.  }
    \label{fig:tau}
  \end{figure}

\begin{figure}
    \centering
    \includegraphics{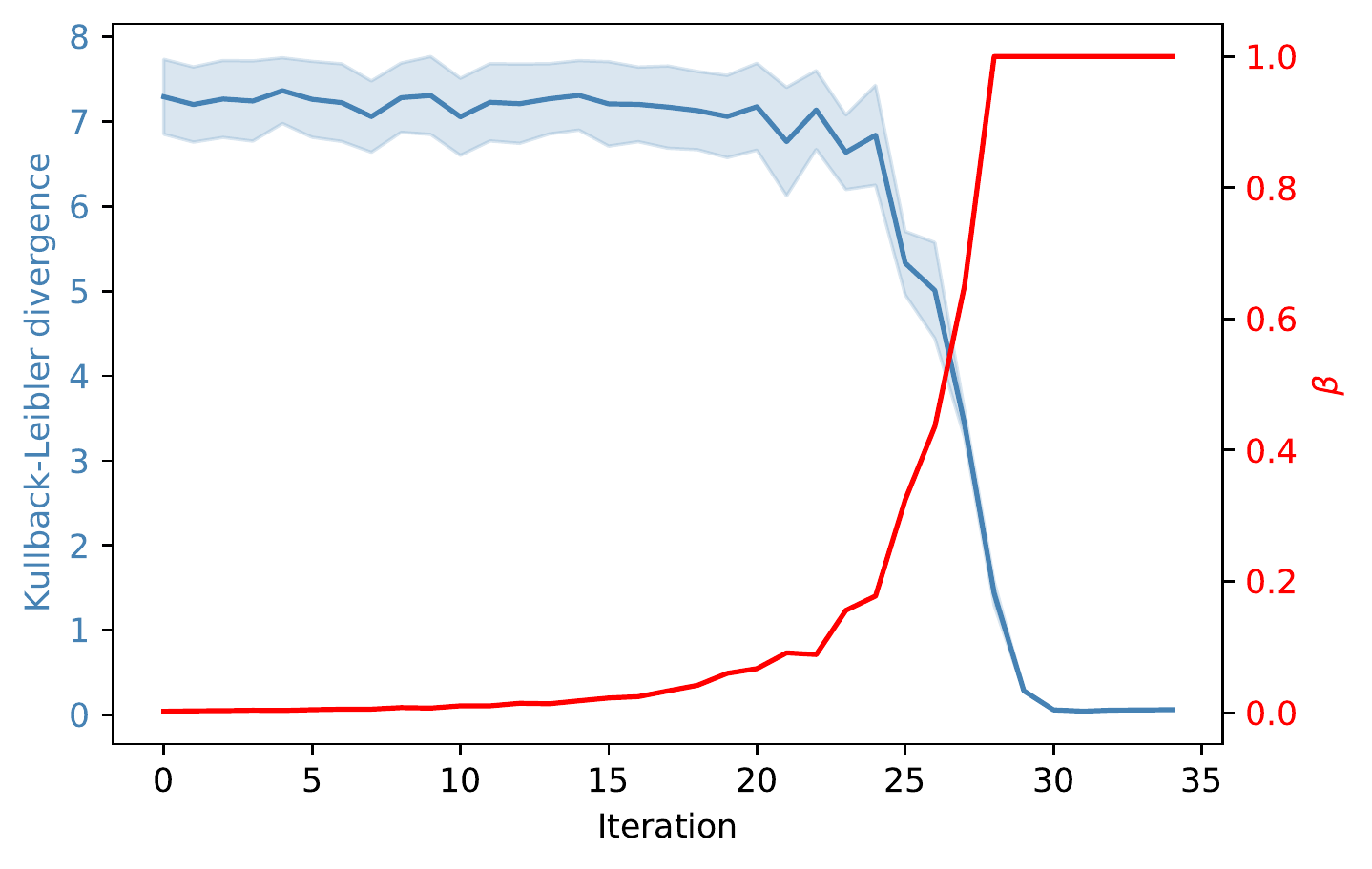}
    \caption{Typical evolution of the inverse temperature $\beta$ ($y$-axis in red) and estimated
      Kullback-Leibler divergence ($y$-axis in blue) along iterations ($x$-axis). The automatically
      calibrated $\beta$ starts by increasing slowly until a sharp acceleration, followed by
      stabilization clearly indicating convergence of sequence of proposal distributions. The
      estimated KL divergence shows the upper bound biais until iteration $20$, as detailed in
      \ref{sec:diagnostics}. Yet its sharp decrease and stabilization mirrors $\beta$'s path.}
    \label{fig:monitoring}
  \end{figure}

\begin{figure}
   \centering
    \includegraphics{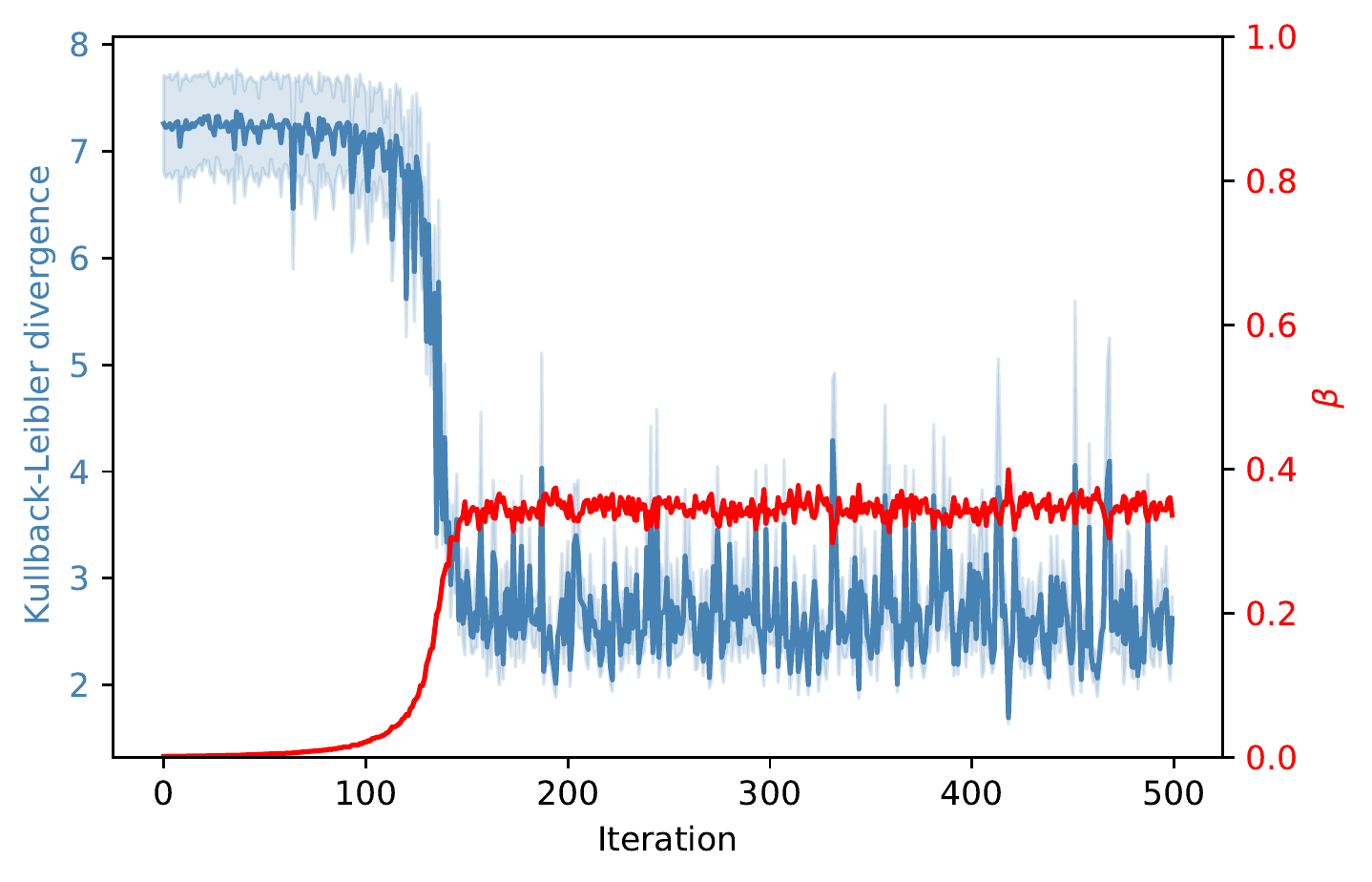}
    \caption{A very high-dimensional problem : The target is a $1000$-dimensional gaussian distribution, the proposals are gaussian distributions with diagonal covariance. (left) Evolution of  the inverse temperature $\beta$ (in red) and estimated Kullback-Leibler divergence (blue) along iterations. \textit{(right)} the L2 distance between the moments of the target and proposal distribution at each iteration. The temperature doesn't go to $1$ despite the target distribution belonging to the family of proposal distributions and the covariance of the proposal doesn't converge to the real covariance. }
    \label{fig:high-dim}
\end{figure}

\section{Numerical Experiments}

We finally illustrate the good numerical properties of TAMIS relatively to its initialization and to
the dimensionality of the problem.

\begin{table}[tb]
  \centering
  \caption{Initialization and dimensionality}
  \label{tab:4}
  \renewcommand{\arraystretch}{1.2}
  \begin{tabular*}{.95\textwidth}{@{\extracolsep{\fill}} l l l l}\toprule
    \textbf{Experiment} & E4.1 & E4.2 & E4.3
    \\ \hline
    \textbf{Dimension} & $d\in\{20,50\}$ & $d\in\{5, 10, 20, 50, 100\}$ &  $d\in\{300, 500\}$
    \\
    \textbf{Target} & Rosenbrock distr. &\multicolumn{2}{c}{$\mathcal N(50, 5)^{\otimes d}$}  
    \\
    \textbf{Proposal} & \multicolumn{3}{c}{Gaussian mixture with 5 components} 
    \\
    \textbf{Draws} & $N_t=2,000$ & $N_t=1,000$ & $N_t=2,000$
    \\
    $\textbf{ESS}_\textbf{min}$ & $100$ & $300$ & $1,000$
    \\
    $\boldsymbol\tau$ & \multicolumn{3}{c}{$0.4$}
    \\
    \textbf{Stop} & $t=20$ &\multicolumn{2}{c}{$\sum_t \ESS_t > 1,000$} 
    \\ \bottomrule
  \end{tabular*}
\end{table}

\subsection{On the effect of initialization}
\label{sec:init}
We now compare the effect of a bad initialization on TAMIS, AMIS and N-PMC with Experiment E4.1
given in Table~\ref{tab:4}. The example considered is the banana shape target density of Haario et
al., also know as the Rosenbrock distribution. Let
$ \sigma^2 = 100, \Sigma = \text{diag}(\sigma^2,1, \dots , 1), b= 0.03$ and
$\Psi(x) =\left(x_1,x_2+b(x_1^2-\sigma^2), x_3, \dots , x_d\right)$. The target is the Rosebrock
distribution with density
\begin{equation*}
    \pi(x) = \varphi( \Psi(x) | 0, \Sigma).
\end{equation*}
For N-PMC, the inverse temperature sequence is chosen as in \cite{Koblents}, i.e., $\beta_t = 1/
(1+e^{-(t-\ell)})$ where $\ell$ is a tuning parameter we have set to $5$.

The first proposal at initialization is a Gaussian mixture model with 5 components with covariance
matrix all equal to $\Sigma$, and centered at random $\mu_k$ drawn from
$\mathcal{N}(0, \Sigma_{0,k}/5)$. We used various covariance matrices $\Sigma$, starting from the
diagonal matrix $\mathrm{diag}(200,50,4,\dots, 4)$ used in \cite{PhysRevD.80.023507} and
\cite{Koblents}. This initial covariance matrix is already adapted to the target and can be
considered as an a priori informed proposal. Then, we used less informed covariance matrices for
$\Sigma$:
\begin{itemize}
\item $\mathrm{diag}(200,50,10,\dots, 10)$, 
\item $\mathrm{diag}(200,50,20,\dots, 20)$,
  $\mathrm{diag}(200,50,50,\dots, 50)$, 
\item $\mathrm{diag}(200,100,100,\dots, 100)$ and 
\item finally  $200 \times \mathrm{I}_d$ which is blind regarding the shape of the target.
\end{itemize}
Each experiment was repeated $500$ times.

Figure~\ref{fig:NPMC} shows the final ESS.  As expected the final ESS
after a fixed number of iterations decreases as the initialization gets worse. Since the dimension
is already high, AMIS fails very frequently even with the first initialization. The tempering scheme
of N-PMC is effective only with a well calibrated initialization, while TAMIS remains effective and
allows the algorithm to converge in every case without any additional parameter tuning.

    \begin{figure}
\centering
\begin{minipage}{.5\textwidth}
   \centering
    \includegraphics[scale=0.5]{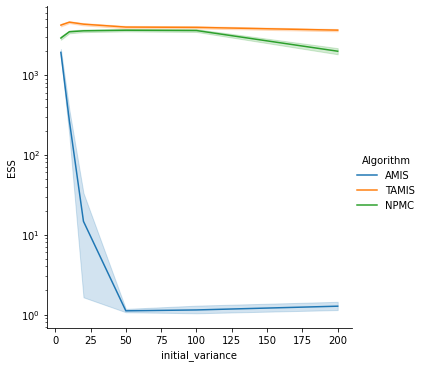}
\end{minipage}%
\begin{minipage}{.5\textwidth}
    \centering
    \includegraphics[scale=0.5]{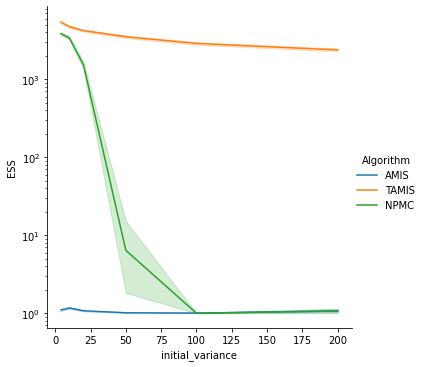}
\end{minipage}

  \caption{Effective Sample Size ($y$-axis) of AMIS, N-PMC and TAMIS after $40,000$ draws along 20
    iterations, with increasingly wide covariance matrix at initialization ($x$-axis) in dimension
    20 (left) and 50 (right). As expected from the litterature, AMIS is only performing well with a
    good initialization and if the dimension is relatively low. N-PMC is able to correct for bad
    initialization with a well chosen tempering path if the dimension is low enough, while TAMIS
    performs well in every case.}
  \label{fig:NPMC}
\end{figure}
\subsection{On the effect of dimensionality}
\label{sec:dim}
We now consider a simple Gaussian target
\begin{equation*}
  \mathcal{N}(50, 5)^{\otimes d}
\end{equation*}
of Experiment E4.3 of Table~\ref{tab:4} in high dimension. We only consider TAMIS only, as both AMIS
and N-PMC fail in every case.  The initialization of the proposal distribution is poor for both
location and for scale. The proposal distributions are Gaussian mixture models with $5$
components. At initialization, they are centered at random
$\mu_{k} \sim \text{Unif}([-4,4])^{\otimes d}$ and have covariance matrices
$\Sigma_{k} = 200\times \mathrm{I}_d$.  The target is therefore very concentrated and centered very
far in the tail of the initial proposal.  The other tuning details are given in Table~\ref{tab:4}.

We plot the MSE when estimating the trace of the covariance matrix along iterations. We also plot the 
number of likelihood evaluations required before convergence of the proposal (assessed by the number
of iterations such that $\hat{\mathrm{KL}}(\pi || q_t)>1$ in Figure~\ref{fig:Dimension}.

The number of simulations required before convergence increases as expected with the dimension. But
we note that not only is TAMIS able to accurately estimate scale and location of a very high
dimensional target, it does so with the same bad initialization as previously, with very little
tuning required.

\begin{figure}
\centering
\begin{minipage}{.5\textwidth}
   \centering
    \includegraphics[scale=0.5]{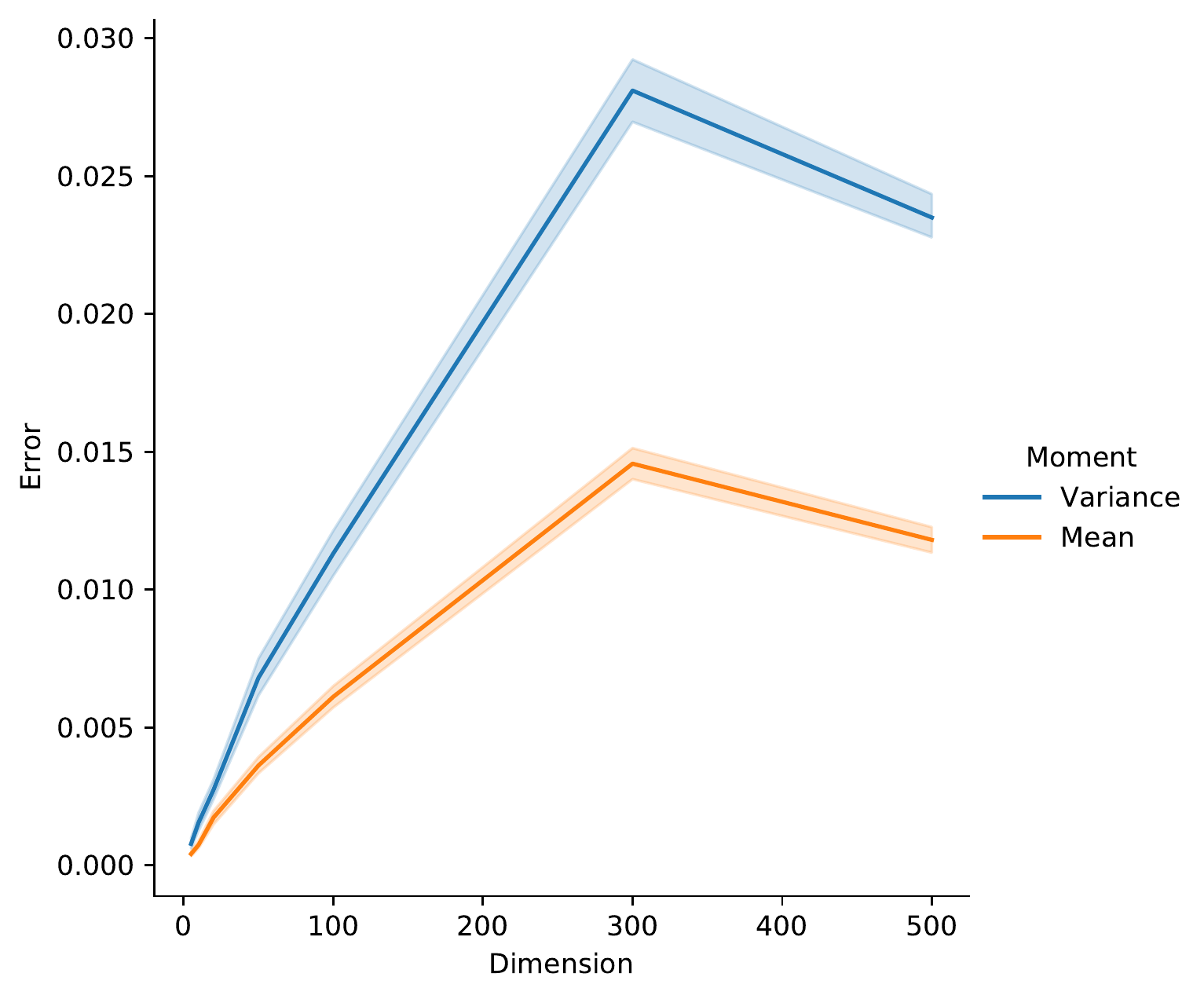}
\end{minipage}%
\begin{minipage}{.5\textwidth}
    \centering
    \includegraphics[scale=0.5]{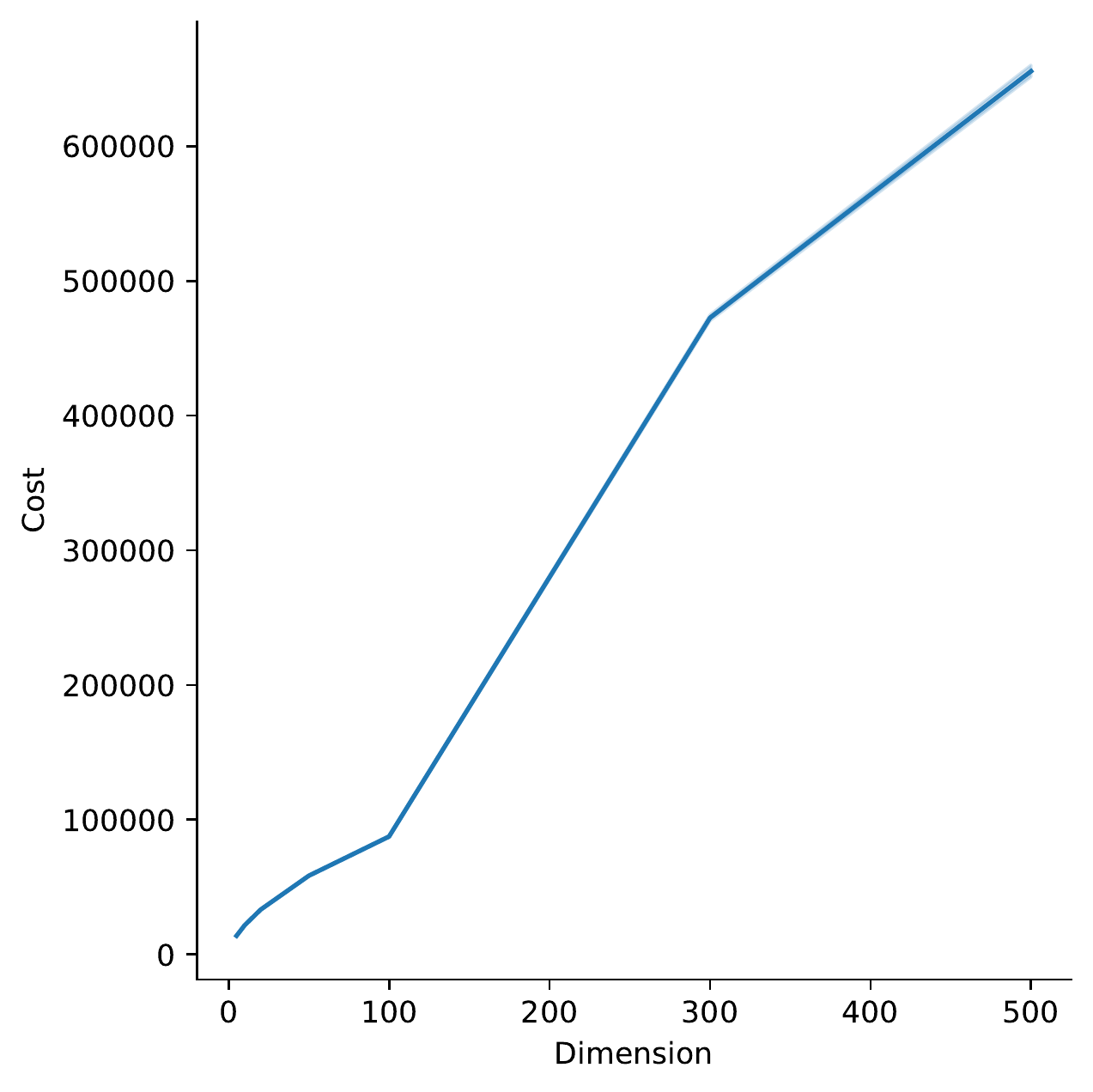}
\end{minipage}

\caption{Mean square error ($y$-axis on the left) of the estimates of the mean and covariance for
  increasing dimension ($x$-axis) and the required number of iterations ($y$-axis on the right)
  before convergence of the proposal to the target distribution (right).}
    \label{fig:Dimension}
\end{figure}

\section{Conclusion}

We have designed an adaptive importance sampling that is 
\begin{itemize}
\item robust to poor initialization of proposal and
\item robust to high dimension of the space to sample
\item efficient in the number of evaluations of the target density and
\item does not rely on any gradient computation.
\end{itemize}
Very few importance sampling algorithm are stable in dimension higher than $100$, and TAMIS is one
of them. Therefore, TAMIS can be used to initialize other Monte
Carlo algorithm such as MCMC methods that can lead to more precise estimates when correctly initialized.
The phase transition observed in the decrease of the Kullbuck-Leibler divergence we monitor remains
to be explained theoretically. 

\bibliographystyle{apalike} \bibliography{TAMIS_draft}


\appendix

\section{Results on the tempered targets}
\label{app:tempering}

Here, we consider that $\pi(x)$ and $q_t(x)$ are normalized densities.

For all $\beta\in[0;1]$, we introduce the normalized density
\begin{align*}
  \pi_{\beta,t}(x) &=
                          \frac{1}{C_t(\beta)}\pi^{\beta}(x)q_t^{1-\beta}(x)
                          \quad \text{where }
  C_t(\beta)=\int \pi^{\beta}(x)q_t^{1-\beta}(x) \dd x.
\end{align*}
Since the logarithm is a concave function, we have for all $\beta$ and
$x$,
\[
  \pi^{\beta}(x)q_t^{1-\beta}(x) \le \beta \pi(x) +
  (1-\beta) q_t(x).
\]
Thus, for all $\beta$, $C_t(\beta)\le 1$. Moreover, $C_t(0)=C_t(1)=1$.

\begin{pro} \label{pro:KLbeta}
  The function $\beta\mapsto \KL(\pi\|\pi_{\beta,t})$ is a convex, non
  increasing function. It decreases from $\KL(\pi|q_t)$ to $0$.
\end{pro}
\begin{proof}[Proof of Proposition~\ref{pro:KLbeta}]
  Set for all $\beta$, $k(\beta)=\KL(\pi\|\pi_{\beta,t})$. We
  have
  \begin{align*}
    k(\beta) & = \int \pi(x) \log
               \frac{\pi(x)C_t(\beta)}{\pi^{\beta}(x)q_t^{1-\beta}(x)}
               \dd x
    = (1-\beta) \KL(\pi\|q_t) + \log C_t(\beta).
  \end{align*}
  Hence its first and second derivatives are
  \begin{align} \label{eq:kprime}
    k'(\beta) = - \KL(\pi\|q_t) + \frac{C'_t(\beta)}{C_t(\beta)},
    \quad
    k''(\beta) = \frac{C''_t(\beta)}{C_t(\beta)} - \left(\frac{C'_t(\beta)}{C_t(\beta)}\right)^2.
  \end{align}

  On the other hand, the first and second derivative of $C_t(\beta)$ are
  \begin{align*}
    C'_t(\beta) &= \int \pi^{\beta}(x)q_t^{1-\beta}(x)
                  \log \frac{\pi(x)}{q_t(x)}\dd x =
                  C_t(\beta) \, \esp_{\beta,t} \left(\log\frac{\pi(x)}{q_t(x)}\right),
    \\
    C''_t(\beta) &= \int \pi^{\beta}(x)q_t^{1-\beta}(x)
                   \log^2 \frac{\pi(x)}{q_t(x)}\dd x =
                   C_t(\beta) \,\esp_{\beta,t} \left(\log^2\frac{\pi(x)}{q_t(x)}\right).
  \end{align*}
  where $\esp_{\beta,t}$ is the expected value when $x\sim\pi_{\beta,t}(x)$.
  Thus, using \eqref{eq:kprime}, 
  \[
    k''(\beta) = \text{Var}_{\beta,t}\left(\log\frac{\pi(x)}{q_t(x)}\right) \ge 0
  \]
  and $k(\beta)$ is a convex function.
  
  Moreover, using \eqref{eq:kprime} again, we have
  \[
    k'(1) = - \KL(\pi\|q_t) + \frac{C'_t(1)}{C_t(1)} =
    - \KL(\pi\|q_t)  + \int \pi(x) \log\frac{\pi(x)}{q_t(x)}\dd x=0.
  \]
  Because of the convexity of $k$, for all $\beta\in[0,1]$, $ k'(\beta)\le k'(1)=0$.  Thus,
  $k(\beta)$ is decreasing and the proof is completed.
\end{proof}

The proposition given below is similar to the one of \cite{Beskos}, but the proof we give here deals
with finite samples.
\begin{pro}
\label{decreasing_ESS}
  Consider a collection of positive weights $w_i$, $i=1,\ldots,n$.
  The function $\beta \mapsto \ESS(\beta)$ defined by
  \[
     \ESS(\beta) =
      \left(\sum_{i=1}^{n} w_i^{\beta} \right)^2 \Big/
    \left(\sum_{i=1}^{n} w_i^{2\beta}\right)
  \]
  is decreasing.
\end{pro}
\begin{proof}
  If $x>0$, the derivate of $x^\beta$ with respect to $\beta$ is $x^\beta\log x$.
  Hence,
  \begin{align*}
    \frac{\dd}{\dd\beta}\ESS(\beta) &=
                                      \frac{\ds 2 g(\beta)\sum_{i=1}^nw_i^\beta}{\ds\left(\sum_{i=1}^n
    w_i^{2\beta}\right)^2}  \quad \text{where}
    \\
    g(\beta) &= \left(\sum_{i=1}^n w_i^{\beta}\log w_j\right)
    \sum_{j=1}^n w_j^{2\beta} - \left(\sum_{j=1}^nw_j^{2\beta}\log
      w_i\right) \sum_{i=1}^n w_i^{\beta}.
  \end{align*}
  Now,
  \begin{align*}
    g(\beta) & = \sum_{1\le i,j \le n}
               w_i^{2\beta}w_j^\beta \Big( \log w_j- \log
               w_i \Big)
    \\
    & = \sum_{1\le i < j \le n} w_i^\beta w_j^\beta  \Big( \log w_j- \log
      w_i\Big)  \Big( w_i^\beta- w_j^\beta \Big)
    \\
    &\le 0,
  \end{align*}
  since, for all $a,b>0$,
  \[
    a^{2\beta}b^\beta (\log b - \log a) +
    a^\beta b^{2\beta}(\log a - \log b) = a^\beta b^\beta \Big(a^\beta
    - b^\beta \Big)\log \frac ba \le 0. \qedhere
  \]
\end{proof}

\section{Proof of Proposition~\ref{pro:KL_hat_pi}}
\label{app:proof}
We start with this simple Lemma.
\begin{lem} \label{lem:Ebar}
  Let $f(x)$ and $g(x)$ be two densities on the $x$-space, which partitioned by
  $E \cup \bar E$.  Introduce the normalized densities knowing $x\in E$ or $\bar E$ as
  \[
    f_{|E}(x) = \frac{1}{f(E)}f(x)\mathbf 1_E(x), \quad
    f_{|\bar E}(x) = \frac{1}{f(\bar E)}f(x)\mathbf 1_{\bar E}(x)
  \]
  and likewise for $g_{|E}$ and $g_{|\bar E}$. We have
  \[
    \KL\big(f \big\| g \big) = f(E)\, \KL\big(f_{|E} \big\| g_{|E} \big) +
    f(\bar E)\, \KL\big(f_{|\bar E} \big\| g_{|\bar E} \big) +
    f(E) \log \frac{f(E)}{g(E)} + f(\bar E) \log \frac{f(\bar E)}{g(\bar E)}.
  \]
\end{lem}
\begin{proof}
  We have
  \[
    \KL\big(f \big\| g \big) = \int_E f(x) \log \frac{f(x)}{g(x)}\dd x +
    \int_{\bar E} f(x) \log \frac{f(x)}{g(x)}\dd x.
  \]
  Moreover
  \begin{align*}
    \int_E f(x) \log \frac{f(x)}{g(x)}\dd x
    & =\int_E f(E) f_{|E}(x) \log \frac{f(E) f_{|E}(x)}{g(E) g_{|E}(x)}\dd x
    \\
    & = f(E)\int_E f_{|E}(x) \log \frac{f_{|E}(x)}{ g_{|E}(x)}\dd x
      +  f(E) \log \frac{f(E)}{g(E)} \\
    & = f(E)\, \KL\big(f_{|E} \big\| g_{|E} \big) +  f(E) \log \frac{f(E)}{g(E)}.
  \end{align*}
  Likewise,
  \begin{align*}
    \int_{\bar E} f(x) \log \frac{f(x)}{g(x)}\dd x
    &= f(\bar E)\, \KL\big(f_{|\bar E} \big\| g_{|\bar E} \big) 
           + f(\bar E) \log \frac{f(\bar E)}{g(\bar E)}. \qedhere                                                              
  \end{align*}
\end{proof}

\begin{proof}[Proof of Propostion \ref{pro:KL_hat_pi}]
  Using Lemma~\ref{lem:Ebar}, $\KL\Big( \pi_{\beta,t}\Big\|\widehat\pi_{\beta,t} \Big) = \KL_{I} + \KL_{II} + \KL_{III}$ where
  \begin{align*}
    \KL_{I} & = \pi_{\beta,t}(E) \KL\Big( \pi_{\beta,t}^E \Big\| q_t^E \Big)
    \\
    \KL_{II} & = \pi_{\beta,t}(\bar E) \KL\Big( \pi_{\beta,t}^{\bar E} \Big\| \pi_{\beta,t}^{\bar E}
               \Big) =0
    \\
    \KL_{III} & = \pi_{\beta,t}(E) \log\frac{\pi_{\beta,t}(E)}{\lambda} + (1-\pi_{\beta,t}(E)) \log \frac{1-\pi_{\beta,t}(E)}{1-\lambda}.
  \end{align*}
  Likewise,
  $\KL\Big( \pi_{\beta,t}\Big\|q_t\Big) = \KL_{I}' + \KL_{II}' + \KL_{III}'$ where
    \begin{align*}
    \KL_{I}' & = \pi_{\beta,t}(E) \KL\Big( \pi_{\beta,t}^E \Big\| q_t^E \Big) = \KL_{I}
    \\
    \KL_{II}' & = \pi_{\beta,t}(\bar E) \KL\Big( \pi_{\beta,t}^{\bar E} \Big\| q_t^{\bar E} 
               \Big) \ge 0 = \KL_{II}
    \\
    \KL_{III}' & = \pi_{\beta,t}(E) \log\frac{\pi_{\beta,t}(E)}{q_t(E)} + (1-\pi_{\beta,t}(E)) \log
                 \frac{1-\pi_{\beta,t}(E)}{1-q_t(E)}.
    \end{align*}
    Moreover, when $s\le 1$, $\lambda = s q_t(E) \le q_t(E)$, thus $\KL_{III}' \ge \KL_{III}$.
    Finally,
    \[
      \KL\Big( \pi_{\beta,t}\Big\|\widehat\pi_{\beta,t} \Big) = \KL_{I} + \KL_{II} + \KL_{III} \le
      \KL_{I}' + \KL_{II}' + \KL_{III}' = \KL\Big( \pi_{\beta,t}\Big\|q_t\Big). \qedhere
    \]
\end{proof}

\end{document}